\documentclass[lettersize,journal]{IEEEtran}
\usepackage{amsmath,amsfonts}
\usepackage{algorithmic}
\usepackage{algorithm}
\usepackage{array}
\usepackage[caption=false,font=normalsize,labelfont=sf,textfont=sf]{subfig}
\usepackage{textcomp}
\usepackage{stfloats}
\usepackage{url}
\usepackage{verbatim}
\usepackage{graphicx}
\usepackage{cite}
\newtheorem{assumption}{Assumption}
\newtheorem{definition}{Definition}
\newtheorem{theorem}{Theorem}
\newtheorem{lemma}{Lemma}
\newtheorem{remark}{Remark}

\newenvironment{proof}{{\noindent\it Proof}\quad}{\hfill $\square$\par}

\begin{document}

\title{Anti-Delay Kalman Filter Fusion Algorithm for Vehicle-borne Sensor Network with Finite-Time Convergence}

\author{Hang Yu, 
	Keren~Dai,~\IEEEmembership{Member, IEEE},
	Qingyu~Li,
	Haojie~Li,
	Yao~Zou,~\IEEEmembership{Member,~IEEE,}
	Xiang~Ma,
	Shaojie~Ma
	and~He~Zhang
	\thanks{\textit{Corresponding author: Keren Dai}.}    
	\thanks{H. Yu, K. Dai, H. Li, X. Ma, S. Ma, H. Zhang are with School of Mechanical Engineering, Nanjing University of Science and Technology, Nanjing, Jiangsu, 210094 China, e-mail:(dkr@njust.edu.cn).}
	\thanks{Q. Li is with the North Information Control Research Academy Group Company, Ltd., Nanjing, Jiangsu, 211153, China.}
	\thanks{Y. Zou is the with School of Automation and Electrical Engineering, University of Science and Technology Beijing, Beijing 100083  P. R. China, and with the Institute of Artificial Intelligence, University of Science and Technology Beijing, Beijing 100083, China.}
\thanks{Manuscript received ---; revised ----.}}

\markboth{Journal of \LaTeX\ Class Files,~Vol.~14, No.~8, August~2021}%
{Shell \MakeLowercase{\textit{et al.}}: A Sample Article Using IEEEtran.cls for IEEE Journals}


\maketitle

\begin{abstract}
Intelligent vehicles in autonomous driving and obstacle avoidance, the precise relative state of vehicles put forward a higher demand.
For a vehicle-borne sensor network with time-varying transmission delays, the problem of coordinate fusion of vehicle state is the focus of this paper.
By the ingeniously designed low-complexity integration with a consensus strategy and buffer technology, an anti-delay distributed Kalman filter (DKF) with finite-time convergence is proposed. 
By introducing the matrix weight to assess local estimates, the optimal fusion state result is available in the sense of linear minimum variance.
In addition, to accommodate practical engineering in intelligent vehicles, the communication weight coefficient and directed topology with unidirectional transmission are also considered.
From a theoretical perspective, the proof of error covariances upper bounds with different communication topologies with delays are presented. 
Furthermore, the maximum allowable delays of vehicle-borne sensor network is derived backwards.
Simulations verify that while considering various non-ideal factors above, the proposed DFK algorithm produces more accurate and robust fusion estimation state results than existing algorithms, making it more valuable in practical applications. Simultaneously, a mobile car trajectory tracking experiment is carried out, which further verifies the feasibility of the proposed algorithm.	
\end{abstract}


\section{Introduction}
\IEEEPARstart{W}{ith} the development of sensing techniques and vehicle communication facilities, some vehicles can interact with vehicle-borne sensors to achieve accurate positioning by using data fusion technology.
Multi-sensor data fusion technology can solve dynamic target tracking problems and is also research hotspot in the field of intelligent vehicles, vehicle network, autonomous driving and obstacle avoidance \cite{sun2017multi,sun2020distributed,dormann2018optimally,8897571,9035421,9378811}.  According to the literature, there are currently two methods for data fusion technology. The first method is a centralized filter \cite{willner1976kalman}, in which the observations of all sensors are transmitted to the fusion central to generate a state estimation with a highest accuracy. However, due to the need to process a large amount of data, this method leads to a severe computational burden. Also, sensor failures will directly affect the final fusion result and produce poor accuracy and stability. 
This feature will seriously affect the safety of vehicles.
The second method is a distributed
filter, in which the local estimator of the sensor can obtain the global optimal or suboptimal state estimate according to a certain information fusion algorithm. The distributed filter exhibits a low complexity and computational burden, and is resistant to sensor failures. This is very important for accurate vehicle positioning and tracking in complex traffic environment.
Thus, this study's design incorporates distributed technologies.

Over the past decades, Kalman filtering (KF) has been widely studied in the fields of target detection, location and tracking, industrial monitoring and signal processing \cite{song2007optimal,sun2004multi,sun2004multi2,hashemipour1988decentralized,shen2012globally,8897571,9705142}. In particular, the DKF  approach  has been developed for many years due to its good performance at dynamic target tracking, which is widely used in the vehicle network \cite{lian2020distributed,marelli2018distributed,ruan2018globally,liu2018distributed,9035421,9050874}.    

To facilitate the accurate control of vehicle state in the process of positioning and tracking, the sensor network is required to have the consistency of target state estimation results. Consensus control has been widely used in the formation of UAVs, cluster satellites and the coordination of mobile robots \cite{xiao2004fast}.
In the field of multi-sensor networks, consensus research has been performed for many years  \cite{2020Distributed,he2018consistent,xiao2004fast,carli2008distributed,das2015distributed}.  There is a common feature in these studies: the average consensus strategy is used. However, this strategy has some disadvantages. First, theoretically, infinite iterations are required to achieve asymptotic convergence. Second, the stop strategy of a multi-sensor system cannot be clearly described in practical applications \cite{wu2018distributed}. 
But for vehicle sensor networks, it is necessary to judge the relative state of vehicles in a short time. Therefore, it requires that the estimation error of the filtering algorithm has the characteristics of finite-time convergence.

To overcome these shortcomings, the DKF algorithm with finite-time convergence is proposed in \cite{wu2018distributed,liu2018distributed,di2015distributed,thia2013distributed}.  Using the maximum consensus technique, a finite-time KF algorithm with local unobservability was proposed in \cite{liu2018distributed}. In \cite{wu2018distributed},  a finite-time DKF algorithm was proposed, which can achieve dynamic monitoring of linear discrete dynamic systems using a sensor network with active sensors and some idle sensors
. In \cite{di2015distributed,thia2013distributed}, using maximum consistency technology, the sensor network can achieve consistency in finite time.

In addition, the necessary interactive connections between multi-sensor systems are typically established through vehicle network. However, in practical applications, transmission delays in vehicle network are inevitable. Due to transmission delays, the filtering accuracy of a multi-sensor system may be reduced, resulting in excessive deviation of the estimation of vehicle state. When error accumulates to a certain extent, it can affect system stability and make the algorithm unusable. Therefore, research on the anti-delay of vehicle network is important and necessary \cite{zou2021sampled}.

To date, many researchers have investigated transmission delays in multi-sensor systems \cite{2020Distributed,xing2016multisensor,dong2019stability,zhang2017leader,jin2021distributed,sun2012optimal,shi2009kalman,chen2014distributed,li2019event}.
Research on transmission delays can be primarily divided into two groups: single-channel and multi-channel sensor network systems. For a single channel sensor network system, \cite{sun2012optimal} designed a linear filter based on the augmented method but exhibited high computational costs. Thus, a nonaugmented filter was proposed using the orthogonal projection method \cite{sun2014modeling}. In \cite{shi2009kalman}, a KF algorithm with a finite length buffer was designed using timestamp technology.
For a multi-channel sensor network system, a DKF algorithm that can deal with transmission delay was proposed in \cite{2020Distributed}, and the algorithm can achieve a consensus of estimation error. In \cite{xing2016multisensor}, considering the existence and correlation of process noise and measurement noise, a DKF algorithm that considers the transmission delays of a sensor network was proposed.
For the transmission delays and packet dropout in multi-channel sensor networks, a fusion estimation method with matrix weighting was designed in \cite{chen2014distributed}.

To the best of the authors' knowledge, in the filed of vehicle network, the aforementioned references do not take into account the finite-time consensus, communication weight and anti-delay of a vehicle-borne sensor network at the same time, while these performance criteria and non-ideal factors are critical for target vehicle precise positioning and tracking. Thus, a DFK algorithm that addresses all these issues would be valuable and challenging because it would be more suitable and feasible for various applications.
The main contributions of this study are as follows:
\begin{enumerate}
\item Inspired by  \cite{wu2018distributed,2020Distributed,liu2018distributed,shi2009kalman,xing2016multisensor,chen2014distributed,shi2009kalman,2020Distributed,li2019event}, we integrate the finite-time control strategy and buffer structure to achieve finite convergence of vehicle state estimation error and anti-delay simultaneously.
The ingenious trick is using distributed design and optimal vector weight technology to address the high complexity problem caused by direct integration.
This makes it possible for the designed algorithm to be practically applied in the vehicle network.
 With this new architecture, the proof of vehicle state estimate error convergence in finite time has been completed.

\item In the case of transmission delay for vehicle-borne sensor network, the upper bound of the error covariance is established by using the inverse matrix theorem and Gramian matrix characteristics. Based on this result,  the maximum possible delay under the allowable detection accuracy can be obtained.

\item Considering the complex traffic environment, the ideal undirected communication topology is difficult to guarantee. The requirement of communication topology is reduced from an undirected graph to a strongly connected digraph. It is proved that under the condition of a digraph, finite-time the convergence of estimation error  and anti-delay performance can be achieved simultaneously.  Furthermore, the upper bound of the estimation error covariance with a strongly connected digraph topology is derived. This result makes the proposed algorithm applicable to more scenarios with worse communication conditions. 

\item Simulations show that the proposed DFK algorithm generates results for fusion estimating state that are more reliable and accurate than those of existing algorithms. A mobile car trajectory tracking experiment is also conducted to confirm the viability of the suggested method.
\end{enumerate}

The remainder of this paper is organized as follows. Section \uppercase\expandafter{\romannumeral 2} describes necessary preliminaries, notations and transmission delay problems. The DKF with finite-time convergence and transmission delays are primarily described in Section \uppercase\expandafter{\romannumeral 3}. Simulation analysis and tracking experiment of the DKF are presented in Section \uppercase\expandafter{\romannumeral5}. Finally, conclusions are given in Section \uppercase\expandafter{\romannumeral6}.

\section{Preliminaries and notations}\label{Problem statement}
\subsection{Problem statement}
In this paper, the distributed vehicle state observation problem of a linear system by a vehicle-borne sensor network with transmission delays is studied. In particular, the observed state evolves from the following discrete-time linear dynamics:

\begin{equation}\label{system_model}
	x_{k+1} = \varPhi _k x_k + w_k
\end{equation}
where $x_k\in \mathbb{R}^{n_x}$ is the target vehicle state to be estimated, $k$ is the step size, $\varPhi_k\in \mathbb{R}^{n_x\times n_x}$ is a known system matrix, and $w_k \in \mathbb{R}^{n_x}$ is the process noise, which is assumed to be Gaussian white noise with zero-mean. 
The sensor network consists of  $n$  sensors, and each sensor has the following measurement equation:
\begin{equation}\label{measurement_equation}
	{y_k^i} = {H_k^i}x_k + {v_k^i}, \qquad i = 1, 2, \dots, n
\end{equation}
where ${y_k^i}\in \mathbb{R}^{m_i}$ is the measurement value of sensor $i$, ${H_k^i}\in \mathbb{R}^{m_i\times n_x}$ is a known measurement matrix, and ${v_k^i} \in \mathbb{R}^{m_i}$ is the  measurement noise, which is also assumed to be Gaussian white noise with zero-mean. In addition, some reasonable assumptions associated with the process $w_k$ and each measurement noise $v_k^i$ are made as follows.   

\begin{assumption}\label{noise_uncorrected}
	The process noise $w_k$ and each measurement noise ${v_k^i}$ are independent. Then:
	\begin{equation}
		\begin{gathered}
			\mathbb{E}[v_t^i{(v_k^j)^T}] = 0, \hfill \\
			\mathbb{E}\left\{ {\left[ {\begin{array}{*{20}{c}}
						{{w_t}} \\ 
						{v_t^i} 
				\end{array}} \right]\left[ {\begin{array}{*{20}{c}}
						{w_k^T}&{{{(v_k^i)}^T}} 
				\end{array}} \right]} \right\} = \left[ {\begin{array}{*{20}{c}}
					{{Q_k}}&0 \\ 
					0&{R_k^i} 
			\end{array}} \right]{\delta _{tk}}, \hfill \\
			{\text{}}i \ne j;\forall t,k, \hfill \\ 
		\end{gathered} 
	\end{equation}
	where $\mathbb{E}(\cdot$) denotes the expectation operator, superscript $T$ denotes the transpose, $Q_k$ and $R_k^i$ are the covariance matrices of $w_k$ and ${v_k^i}$, respectively, and $\delta_{tk}$ is the Kronecker function, which has the following form:
	\begin{equation}\label{Kronecker function}
		\delta_{tk} =
		\begin{cases}
			1,  & {\rm{if}} \ k=t\\
			0,  & {\rm{if}} \ k\neq t
		\end{cases}
	\end{equation}
\end{assumption}

\begin{assumption}\label{initial state}
	The initial state $x_0$ is independent of $w_k$ and $v_k^i$. Then:
	\begin{equation}
		{\mathbb{E}}\left[ x_0 \right] = \mu_0 \nonumber
	\end{equation}
	\begin{equation}
		{\mathbb{E}}\left[ \left(x_0- \mu_0\right)\left(x_0- \mu_0\right)^T \right] = P_0 \nonumber
	\end{equation}
\end{assumption}

\begin{assumption}\label{boundnessed}
	Systems \eqref{system_model} and \eqref{measurement_equation} are uniformly and completely observable, and then there is a Gramian matrix $M_{k+\bar{n},k}$,  which satisfies the following inequality for any instant $k$:
	\begin{equation}\label{completely_observable}
		\alpha I_{n} \leq M_{k+\bar{n},k} = \sum_{l=k}^{k+\bar{n}}O^T_{l,k} \varPhi^T _l R^{-1}_l \varPhi_l O_{l,k}  \leq \beta I_n
	\end{equation}
	where $\alpha$ and $\beta$ are positive constants,  $\bar{n}$ is a positive integer, matrix $O_{l,k} $ has the form of $O_{k,k} = I_n$, $O_{k+1,k} = \varPhi_k$, and $O_{l,k} = \varPhi_{l-1} \cdots  \varPhi_{k}$ for $l\geq k+1$, $\varPhi_l = \textrm{col} \{\varPhi_{l,1}, \dotsm,  \varPhi_{l,n}\}$  is a column block matrix, and $R_l = \textrm{diag}\{R_{l,1}, \dotsm,  R_{l,n}\} $  is a diagonal block matrix. And, let $\bar{\alpha} I_{n} \leq M_{k+\bar{n}+n-1,k} \leq \bar{\beta} I_n$, where $\bar{\alpha}$ and $\bar{\beta}$ are larger than $\alpha$ and $\beta$, respectively.
\end{assumption}

Given a vehicle-borne sensor network, the communication topology among sensors is described by an undirected graph, $\mathcal{G}=(\mathcal{V},\mathcal{E},\mathcal{W})$, which consists of a node set $\mathcal{V} =\{1,2,\dots, n \} $, an edge set $\mathcal{E}\subseteq~\mathcal{V}\times\mathcal{V}$,
and a weighted matrix $\mathcal{W}=[\omega_{ij}]\in\mathbb{R}^{n \times n}$, and the elements of matrix $\mathcal{W}$ are nonnegative (i.e., $\omega_{ij} \geq 0$,$i,j \in \mathcal{V}$).
For an undirected graph $\mathcal{G}$, $(i,j)\in \mathcal{E}\Leftrightarrow(j,i)\in \mathcal{E}$, that is, nodes $i$ and $j$ can sense each other (i.e., $\omega_{ij}=\omega_{ji}$). Graph $\mathcal{G}$ characterizes the communication topology among sensors and is connected if there exists a path involving all nodes.
If a connected undirected graph has no cycle, it is a tree graph. The set of neighbor sensors connected to sensor $i$ is denoted by $\mathcal{N}_i = \{j \in  \mathcal{V}| (i,j) \in \mathcal{E}\}$, and every sensor $i$ collects a measurement $y_i(k)$ at time $k$. The longest path between the two sensors is the diameter $d_{g}$ of graph $\mathcal{G}$ \cite{wu2018distributed,yu2021distributed}.
\begin{assumption}\label{varPhi_initial_cov_invertible}
	\begin{enumerate}
		\item The system matrix $\varPhi_{k}$  is invertible for any $k \in N^+$, and satisfies $\varPhi_{k}^{-1} \geq \eta I$, where $\eta$ is a positive constant;	
		\item Initial covariance $[P_0]^{-1} \geq 0$.
	\end{enumerate}	
\end{assumption}

\begin{lemma}\label{lemma1_matrix_invert_beta}
	\cite{battistelli2014kullback,2020Distributed} Consider the following equality:
	\begin{equation}\nonumber
		\varUpsilon(\varOmega) = (\varPhi \varOmega \varPhi ^T +Q)^{-1}
	\end{equation} 
	If $\varPhi$ is invertible, then the following results hold:
	\begin{enumerate}
		\item If there exits a positive constant $\hat{\gamma} \leq 1$, then the inequality $\varUpsilon(\varOmega) \geq \hat{\gamma} \varPhi^{-T}\varOmega \varPhi ^{-1}$ holds for any $\varOmega \leq \hat{\varOmega} $, where $\hat{\varOmega} $ is any positive semidefinite matrix.
		\item If there exits a positive constant $\breve{\gamma} < 1$, then the inequality $\varUpsilon(\varOmega) \geq \breve{\gamma} \varPhi^{-T}\varOmega \varPhi ^{-1}$ holds for any $\varOmega \geq \breve{\varOmega} $, where $\breve{\varOmega} $ is any positive semidefinite matrix.
	\end{enumerate}	
\end{lemma}
\subsection{buffer technology}
In actual vehicle usage scenarios, subject to the poor sensor performance and the delays influence of external complex traffic environment interference, information time-varying transmission delays inevitably exist in vehicle-borne sensor network. 
Each time-varying transmission delay is represented by $\{d^t_{ij}(k)|d^t_{ij}(k) \in \{0,1,\dots,d_{t}\}, i \in\mathcal{V}, j\in \mathcal{N}_i, k\in N^+\}$, where $d_{t}$ denotes the maximum delay step. 
According to the information transmission process of sensor network, let the data transmission instant is $s$, which satisfies $s \in D_t(k)$,  where $D_t(k)= \{0,\dots,k \}$ for $k \in \{0,\dots, d_t+1\}$, and $D_t(k)= \{k-d_t,\dots,k \}$ for $k \in \{d_t+2,\dots, +\infty\}$. 
To solve the concerned observation problem subject to transmission delays, the transmitted information packets are attached with time stamps such that each sensor is able to identify the transmission instant of each received packet. 
The buffer technology inspired by \cite{shi2009kalman,2020Distributed} is introduced to store the transmitted packets. 
The buffer structure of sensor $i$ is shown in Fig. \ref{fig_buffer}, there are $(|\mathcal{N}_i|-1) \times L$ blocks with $|\mathcal{N}_i|-1$ as the neighbor number of sensor $i$, and $L$ as the buffer length.

The buffer of sensor $i$ has the following working steps at instant $k$:
\begin{enumerate}
	\item Discard the earliest data, i.e., those $(j_1, k-L), \dots,  (j_{|\mathcal{N}_i|-1}, k-L)$.
	\item Store the latest data from the neighbors at time $k$, i.e., those in  blocks $(j_1, k), \dots,  (j_{|\mathcal{N}_i|-1}, k)$.
\end{enumerate}

Moreover, introduce a variable $\gamma^{i,j}_{k}(s)$ to characterize the packet transmission process. In particular, if the packet transmitted from sensor $j$ is received by sensor at or before instant $k$, then, $\gamma^{i,j}_{k}(s)=1$,  otherwise, $\gamma^{i,j}_{k}(s)=0$. Note that, $\gamma^{i,i}_{k}(s) \equiv 1$.

\begin{figure}[!t]
	\centering
	\includegraphics[width=3.2in]{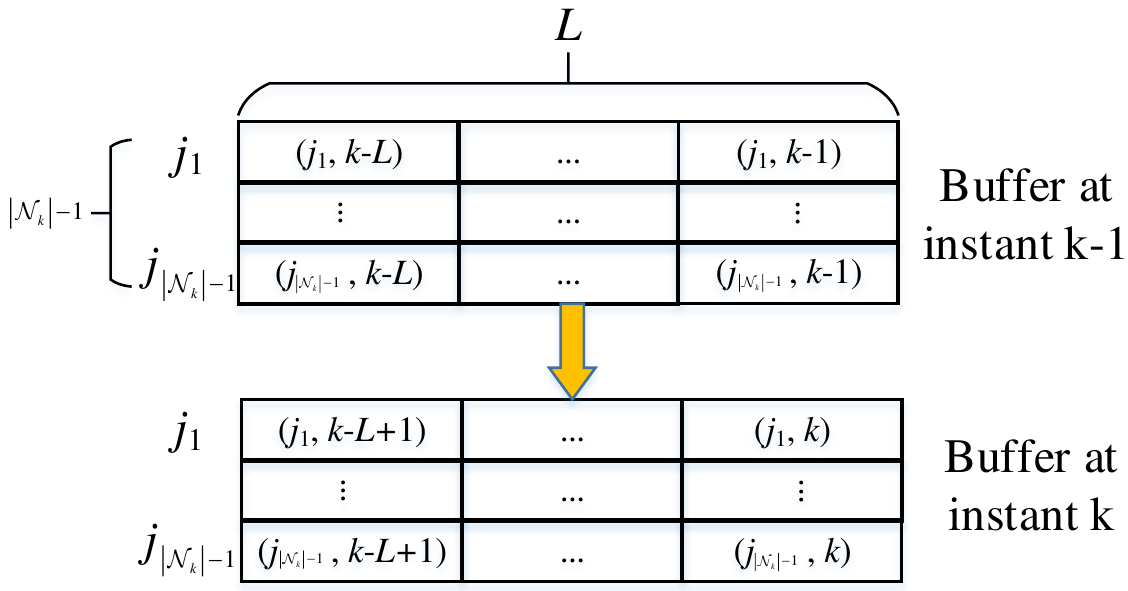}
	\caption{Buffer of sensor $i$.}
	\label{fig_buffer}
\end{figure}

Throughout this paper, the notations are standard. $\mathbb{R}^{n_x}$ and  $\mathbb{R}^{m_i\times n_x}$ stand for the $n_x$-dimensional Euclidean space and set of all ${m_i\times n_x}$ matrices, respectively. The superscript $T$ stands for the transpose of the matrix. The superscript $-1$ stands for the inverse of the matrix. The symmetric matrices $A$ and $B$, where $A \leq B$ stands for  $A-B$ are positive-definite matrices. The notation ${\mathbb{E}}\{\cdot\}$ stands for expectation of variable. $|\cdot|$  stands for the cardinal number of set. $\textrm{diag}\{\cdot\}$ stands for the diagonal matrix element. $\hat{x}_k^i(\cdot)$ and ${x}_k^i(\cdot)$ stand for the prior and posterior estimates of state $x_k$ by sensor $i$, respectively.

\section{Anti-delay DKF algorithm}\label{DKF}

In this section, the anti-delay DKF algorithm is introduced, which has the properties of finite-time convergence and robustness with transmission delays.
The communication weight coefficient and the fusion matrix weight are also considered to accommodate practical engineering, then, it can realize the tracking and positioning of the target vehicle in a more realistic scene.

\begin{algorithm} 
	\caption{Anti-delay DKF} 
	\label{alg1:Framwork} 
	\begin{algorithmic}
		\REQUIRE 
		\STATE $x_{1}^i(0)= \mu_0$, $P_{1}^i(0)=P_0$.	
		\ENSURE 	
		\STATE \textbf{Step 1:} \textbf{If} $k=1$  \textbf{then} \\
		Calculate ${x}{_{1}^i}(1)$ and error covariance  ${P}{_{1}^i}(1)$ by FtDKF, when $k=1$, $s=1$.\\
		\textbf{else if} $2 \leq k \leq d_t+1$ \textbf{then} \\
		\   Let  
		\begin{equation}
			{x}{_{k }^i}(1) =  {x}{_{k-1 }^i}(1)
		\end{equation}
		\begin{equation}
			{P}{_{k }^i}(1) =  {P}{_{k-1 }^i}(1)
		\end{equation}
		\FOR {$s= 2, \dots ,k$}
		\item 
		$\left[{x}{_{k }^i}(s), {P}{_{k }^i}(s)\right] = {\text{FtDKF}}\left[ {x}{_{k }^i}(s-1), {P}{_{k }^i}(s-1), { {w}}_k^{i,j}(s) \right]$
		\ENDFOR \\
		\textbf{else if} $ k \geq d_t+2$ \textbf{then} \\
		\ 	Let  
		\begin{equation}\label{k-d_x}
			{x}{_{k}^i}(k-d_t) =  {x}{_{k-1 }^i}(k-d_t)
		\end{equation}
		\begin{equation}\label{k-d_P}
			{P}{_{k}^i}(k-d_t) = {P}{_{k-1}^i}(k-d_t)
		\end{equation}
		\FOR {$s= k-d_t+1, \dots ,k$}
		\item 
		$\left[{x}{_{k }^i}(s), {P}{_{k }^i}(s)\right] = {\text{FtDKF}}\left[ {x}{_{k }^i}(s-1), {P}{_{k }^i}(s-1), { {w}}_k^{i,j}(s) \right]$
		\ENDFOR\\
		\STATE \textbf{Step 2:} Optimal weighted fusion\\
		\FOR {$k \geq 1$}
		\item 
		\begin{equation}\label{code2_fusion}
			{x}^f_{k}= \varGamma_1 {x}_{k}^1(k)+ \varGamma_2 {x}_{k}^2(k)+\dots+\varGamma_n {x}_{k}^n(k)
		\end{equation}
		\ENDFOR 	
	\end{algorithmic}
\end{algorithm}

${x}_{k}^i(s)$, ${P}^i_{k}(s)$ and $x^f_{k}$ are defined as the local posterior estimation, estimate error covariance and optimal fusion estimate of target vehicle state $x_k$ by sensor $i$ at instant $k$, respectively.

The presented anti-delay DKF algorithm is summarized as Algorithm \ref{alg1:Framwork}.
A finite-time distributed Kalman filter (FtDKF) algorithm, which is shown in Algorithm \ref{alg2:Framwork}.

\begin{algorithm}  
	\caption{\ FtDKF} 
	\label{alg2:Framwork} 
	\begin{algorithmic}
		\STATE \textbf{Step 1:} let $t=1$, calculate $\psi _{s, t-1}^{i\rightarrow j}  = {(H_s^i)^T}{(R_s^i)^{ - 1}}y_{s}^i$, $ \varphi _{s, t-1}^{i\rightarrow j} = {(H_s^i)^T}{(R_s^i)^{ - 1}}H_s^i$. \\
		\STATE \textbf{Step 2:} Transmit  $\psi _{s, t-1}^{i\rightarrow j} $ and $\varphi _{s, t-1}^{i\rightarrow j} $ to sensor $j$, where $j\in  \mathcal{N}_i$.	
		\FOR{ After iteration $t=1,2, \dots, d_{g}$, for each sensor $i$, and $(i,j) \in \mathcal{E} $ }
		\item
		\begin{equation}\label{code3:measure_sum}
			{\Theta _i}(t) = {(H_s^i)^T}{(R_s^i)^{ - 1}}y_s^i + \sum\limits_{j \in \mathcal{N}_i } {{ {\omega}}^{ij}_{k}(s)}  \psi _{s,t-1}^{j \rightarrow i}
		\end{equation}		
		
		\begin{equation}\label{code3:para_sum}
			{\Omega _i}(t) = {(H_s^i)^T}{(R_s^i)^{ - 1}}H_s^i + \sum\limits_{j \in \mathcal{N}_i } {{ {\omega}}^{ij}_{k}(s)}   \varphi _{s,t-1}^{j \rightarrow i}
		\end{equation}
		\STATE $\quad$ For each sensor $j\in \mathcal{N}_i$,
		\begin{equation}
			\begin{split}
				\psi _{s,t}^{j \rightarrow i}& = {\Theta _i }(t) - \psi _{s, t-1}^{i \rightarrow j}\\
				\varphi _{s,t}^{j \rightarrow i}& = {\Omega _i }(t) - \varphi _{s, t-1}^{i \rightarrow j}
			\end{split} 
		\end{equation}     	
		\ENDFOR
		\STATE \textbf{Step 3:} Calculate ${x}{_{k}^i}(s)$ and ${P}{_{k}^i}(s)$ by\\
		\begin{equation}\label{COV_delay}
			{\left[{P}_{k}^i(s)\right]^{ - 1}} = {\left[{P}_{k}^i(s|s-1)\right]^{ - 1}} + {\Omega _i }(d_{g})
		\end{equation}
		\begin{equation}\label{COV_state}
			\begin{split}
				{x}{_{k}^i}(s)  =  {P}{_{k}^i}(s) \bigg[ {{\left({P}_{k}^i(s|s-1)\right)}^{- 1}}{x}_{k}^i(s|s-1)
				+ {\Theta _i }(d_{g}) \bigg]
			\end{split}
		\end{equation}\\
		where 
		\begin{equation}\label{predict_delay_x}
			{x}_{k }^i(s|s-1) = \varPhi_{s-1}  {x}_{k }^i(s-1) 
		\end{equation}
		\begin{equation}\label{predict_delay_P}
			{P}_{k }^i(s|s-1) = \varPhi_{s-1}  {P}_{k }^i(s-1) {\varPhi_{s-1} ^T} + Q_{s-1}
		\end{equation}
		
	\end{algorithmic}
\end{algorithm}

Moreover, the vehicle-borne sensor network communication weight with delay ${{ {\omega}}^{ij}_{k}(s)}$ is denoted as:

\begin{equation}\label{commu_w_delay}
	{{{\omega}}^{ij}_{k}(s)} =
	\begin{cases}
		\gamma^{i,j}_{k}(s){\omega}_{ij},  & {i \neq j} ,\\
		1,  & {i=j}
	\end{cases}
\end{equation}

\begin{theorem}\label{theor_1}
	Considering the vehicle-borne sensor systems \eqref{system_model} and \eqref{measurement_equation}, assuming that the communication topology $\mathcal{G}$ is an undirected tree graph with diameter of $d_g$, the local estimate  ${x}{_{k}^i}(s)$ can be obtained by \eqref{COV_state} in Algorithm \ref{alg2:Framwork}.
\end{theorem}

\begin{proof}\label{proof_of_theor1}
	For the classical local Kalman filter equation \cite{cattivelli2008diffusion,shen2012globally}, we have the following steps:
	
	Update:
	\begin{equation}\label{classic_update_x}
		x^{i}_{k} (s)= 	\hat{x}^{i}_{k}(s)+K^{i}_k (y^{i}_k - H^{i}_k \hat{x}^{i}_{k}(s))
	\end{equation}
	\begin{equation}\label{classic_update_p}
		P^{i}_{k}(s)=\hat{P}^{i}_{k}(s)-K^{i}_k  H^{i}_k \hat{P}^{i}_{k}(s)
	\end{equation}
	\begin{equation}\label{classic_update_k}
		K^{i}_k =  \hat{P}^{i}_{k}(s) (H^{i}_k)^T [H^{i}_k \hat{P}^{i}_{k}(s) (H^{i}_k)^T +R^{i}_k]^{-1}
	\end{equation} 
	
	Prediction:
	\begin{equation}\label{classic_predict_x}
		\hat{x}_{k+1}^{i}(s) = \varPhi_{k} x^{i}_{k}(s)
	\end{equation} 
	\begin{equation}\label{classic_predict_p}
		\hat{P}^{i}_{k+1}(s) = \varPhi_{k} P^{i}_{k}(s)\varPhi_{k}^T +Q_{k-1}
	\end{equation} 
	where $\hat{x}_{k+1}^{i}(s) $ and $\hat{P}^{i}_{k+1}(s)$ are priori state estimator and estimate error covariance, respectively.
	
	Inspired by \cite{cattivelli2010diffusion}, by using the matrix inversion lemma for \eqref{classic_update_x} and \eqref{classic_update_p}, an alternative form  can be obtained:
	\begin{equation}\label{update_x_information}
		x^i_{k}(s)= P^i_{k}(s)\left([\hat{P}^{i}_{k}(s)]^{-1} \hat{x}^i_{k}(s)+\sum_{i \in \mathcal{N}_i} {(H_k^i)^T}{(R_k^i)^{ - 1}}y_k^i    \right)
	\end{equation}
	\begin{equation}\label{update_p_information}
		[P^i_{k}(s)]^{-1} = [\hat{P}^{i}_{k}(s)]^{-1}+ \sum_{i \in \mathcal{N}_i} {(H_k^i)^T}{(R_k^i)^{ - 1}}H_k^i
	\end{equation}
	where the  definitions of $\hat{x}^i_{k}(s)$ and  $ \hat{P}^{i}_{k}(s) $ are consistent with \eqref{classic_predict_x} and \eqref{classic_predict_p}, respectively.
	
	In this paper, since the vehicle-borne sensor network communication weight is considered, the equation of \eqref{update_x_information}  further follows that:
	
	\begin{equation}
		x_{k}^i(s) = P_{k|k}^i\left( { [\hat{P}^{i}_{k}(s)]^{-1} \hat{x}^i_{k}(s) + {\Theta _i}(d_g)} \right)
	\end{equation}
	where ${\Theta _i}(d_g) =  \sum_{j \in {N_i}  } {{w}}_k^{i,j}(s)  \psi _{s,d_g-1}^{j \rightarrow i}$. After $d_g$ information transmissions, $\psi _{s,d_g-1}^{j \rightarrow i}$ is the result produced by $\psi _{s,0}^{j \rightarrow i}$, and $ \psi _{s,0}^{j \rightarrow i} = {(H_s^i)^T}{(R_s^i)^{ - 1}}y_s^i$.
	
	Similarly, by considering \eqref{update_p_information}, it  follows that:
	\begin{equation}
		{[P_{k}^i(s)]^{ - 1}} = [\hat{P}^{i}_{k}(s)]^{-1}  + {\Omega _i}(d)
	\end{equation}
	where ${\Omega _i}(d_g) = {(H_s^i)^T}{(R_s^i)^{ - 1}}H_s^i + \sum\limits_{j \in {N_i}} {w_k^{i,j}}(s) \varphi _{s,d_g-1}^{j \rightarrow i}$.
	
 The equation of \eqref{code2_fusion} is the optimal fusion step, where $\varGamma \in \mathbb{R}^{n_x n\times n_x}$ is the optimal matrix weight and satisfies
	\begin{equation}\label{matrix_weighted} 
		\varGamma= \varXi^{-1}e (e^T\varXi^{-1}e)^{-1} 
	\end{equation}
	where $\varXi = (P^{ij}_{k}) \in \mathbb{R}^{n_x n\times n_x n}$.  $P^{ij}_{k}$ is the local filtering  error cross-covariance between the $i$-th and $j$-th sensor. $e=[I_n,\dots, I_n]^T\in \mathbb{R}^{n_x n\times n_x}$. This value can be obtained by 
	\begin{equation}
		P^{ij}_{k}=({I_{n_x}} - K_{k}^i {H^i_s}) (\varPhi_{k-1} P^{ij}_{k-1}\varPhi_{k-1}^T +Q_{k-1})({I_{n_x}} - K_{k}^j {H^j_s})
	\end{equation}
	where $i,j \in \mathcal{V}$, the error variance of the optimal fusion estimate of state is $P^f_{k|k}=(e^T\varXi^{-1}e)^{-1} $, and $P^f_{k} \leq P^i_{k}$. 
	Proof and more details of $\varXi$ and $\varGamma $ can be obtained directly from \cite{sun2004multi}. 	
\end{proof}

\begin{remark}
	Since the vehicle runs relatively fast, it is convenient to achieve accurate tracking of the target vehicle in a short period of time. Due of the limited computer capability of the vehicle, this requires the algorithm has low complexity.
	It will be computationally burdensome due to the high dimension of matrix $\varXi^{-1}$. \eqref{matrix_weighted} is modified to the optimal vector weight, where $\varXi^{-1} =  (\tilde{P}^{ij}_{k}) \in \mathbb{R}^{n_x n\times n_x n}$ can be used, and $\tilde{P}^{ij}_{k}$ are the diagonal matrices consisting of the diagonal elements of the covariance matrices ${P}^{ij}_{k}$ can be used to reduce the amount of calculation and increase real-time.
	Similarly, the error variance of the optimal fusion estimate of state is $\tilde{P}^f_{k}= (e^T\varXi^{-1}e)^{-1}e^T \varXi^{-1}({P}^{ij}_{k})_{nl \times nl}\varXi^{-1}e(e^T\varXi^{-1}e)^{-1}$, $i,j =1,2,\dots,n$.
\end{remark}

\section{Stability analysis}\label{Stability analysis}
\subsection{Convergence analysis}

The convergence of Algorithm \ref{alg1:Framwork} is analyzed in this subsection. The proposed algorithm is shown to have finite-time convergence. This means that the target vehicle state  estimation error converges for a finite number of iterations

For Algorithm \ref{alg1:Framwork}, it has the following theorem:
\begin{theorem}\label{algorithm_convergence}
	Considering the vehicle-borne sensor systems \eqref{system_model} and \eqref{measurement_equation}, we assume that the graph $\mathcal{G}$ is a connected undirected tree graph. Algorithm \ref{alg1:Framwork} shows the DKF with anti-delay, which has the property of finite-time convergence.
\end{theorem}

\begin{proof}\label{proof_of_theorem1}
	The vehicle-borne sensor connection network graph $\mathcal{G}$ can be separated into two tree sub-graphs $\mathcal{G}_i$ and $\mathcal{G}_j$ connected by sensor nodes ($i$, $j$) because it is an undirected tree graph. The root nodes of the sub-graphs $\mathcal{G}_i$ and $\mathcal{G}_j$ are sensor nodes $i$ and $j$, respectively.
	
	The following procedures must be followed to achieve finite-time convergence of the sensor network's predicted goal state:
	
	\begin{enumerate}
		\item The number of hops from the sensor node to the root node $i$ is defined as the layers of the sensor node in sub-graph $\mathcal{G}_i$. 
		Node $i$ is referred to as a layer-0 node; nodes one hop away from the root node $i$ are referred to as layer-1 nodes, etc.
		\item When $t = 1$, from \eqref{code3:measure_sum}, \eqref{code3:para_sum} and the initial step of Algorithm \ref{alg1:Framwork}, only the information of layer-0 nodes in graph $\mathcal{G}_i$ is used to estimate the target state $x(s)$. The information from other sensor nodes in sub-graph $\mathcal{G}$ is not  used.
		\item After $d_g$ iterations,  $\psi _{s,d_g}^{i \rightarrow j}$ and $\varphi _{s,d_g}^{i \rightarrow j}$ with global information can be acquired utilizing sensor node information from layer-0 to layer-$d_g$ in the sub-graph $\mathcal{G}_i$, and then precise target state estimation can be produced.
		\item Similarly, in sub-graph $\mathcal{G}_j$, because the sensor node $i$ is adjacent to the node $j$, the sensor node information from layer-0 to layer-($d_{g}-1$) is used to obtain $\psi _{s,d_g-1}^{j \rightarrow i}$ and $\varphi _{s,d_g-1}^{j \rightarrow i}$ for the target state estimate after $d_{g}-1$ iterations.
		\item Fusing the information in sub-graph $\mathcal{G}_i$ and $\mathcal{G}_j$, yields
		\begin{equation}\label{information_fusion}
			\begin{split}
				{\Theta _i}(d_g)& = \psi _{s,d_g}^{i \rightarrow j}+ \psi _{s, d_g-1}^{j \rightarrow i}\\
				{\Omega _i}(d_g)& =  \varphi _{s,d_g}^{i \rightarrow j} + \varphi _{s, d_g-1}^{j \rightarrow i}
			\end{split} 
		\end{equation}  
		\item Algorithm \ref{alg1:Framwork} requires $k-s + 1$ iterations at each instant $k$ due to the presence of transmission delays. Based on  the definition of transmission time $s$, Algorithm \ref{alg1:Framwork}  needs a maximum of $d_t(k-s + 1)$ iterations in each instant $k$.
		\item Finally, by substituting ${\Omega _i}(d_g)$ and ${\Theta _i}(d_g)$ in \eqref{information_fusion} into \eqref{COV_delay} and \eqref{COV_state}, respectively, the estimation error covariance and the posterior estimate of the target state can be obtained. 
	\end{enumerate}

	Overall, according to this analysis, the global information of the vehicle-borne sensor network can be used to estimate the target state over a finite number of iterations. To complete finite-time convergence, at least $d_td_g$ iterations are required, which is also the minimum number of iterations needed to achieve finite-time convergence.
\end{proof}

\subsection{Boundedness of Error Covariances}\label{undirected_boundness}

The estimation error covariance $P_{k}^i(s)$ is bounded because the proposed DKF technique can achieve finite-time convergence.

The following theorem is used to demonstrate the boundedness of error covariance:
\begin{theorem}\label{boundedness of error covariance}
	Considering the finite-time DKF algorithm proposed in Algorithm \ref{alg1:Framwork}, when Assumptions \ref{boundnessed} and \ref{varPhi_initial_cov_invertible} exist, the communication graph $\mathcal{G}$ of the vehicle-borne sensor network is undirected and connected, and the delay $d_t$ is bounded. Then, there is a positive constant $\vartheta$, such that the estimate error covariance has $[ {P}_{k}^i(s)]^{-1} \geq \vartheta I$ holds for any $k \in  [1, +\infty]$ and $s \in D_t(k)$.
\end{theorem}

\begin{proof}\label{proof_boundedness}
	The proof of the boundedness of estimate error covariance can be divided into two cases due to the existence of information transmission delays in a vehicle-borne sensor network, that is, $k \in \{1, \dots , (Z+1)d_t-1 \}$ and $k \in \{(Z+1)d_t, \dots , +\infty \}$, where $Z= n+\bar{n}$. In the initial stage, if the transmission delays $d_t$ is a bounded constant and Assumption \ref{varPhi_initial_cov_invertible}  exist, there is $[ {P}_{k }^i(s)]^{-1} \geq \vartheta I$ holds for any $k \in \{1, \dots , (Z+1)d_t-1 \}$ and $s \in D_t(k)$. Also, we only need to prove that $[ {P}_{k }^i(s)]^{-1} \geq \vartheta I$ holds for any $k \in \{(Z+1)d_t, \dots , +\infty\}$ and $s \in D_t(k)$.
	
	From \eqref{COV_delay} and $s=k-d_t$, we have:
	\begin{equation}\label{bounded_all}
		\begin{split}
			&{\left[ {P}_{k }^i(k-d_t)\right]^{ - 1}} = {\left[ \hat{P}_{k }^i(k-d_t|k-d_t-1)\right]^{ - 1}} \\
			& \qquad + \sum\limits_{j \in {N_i}} {{ {\omega}}^{ij}_{k}(k-d_t)} {(H_{k-d_t}^j)^T}{(R_{k-d_t}^j)^{ - 1}}H_{k-d_t}^j
		\end{split}
	\end{equation}
	It can be yielded from \eqref{predict_delay_P} and Lemma \ref{lemma1_matrix_invert_beta} that: 
	\begin{equation} \label{bounded_invertible_trans}
		\begin{split}
			&\left[	 \hat{P}_{k}^i(k-d_t|k-d_t-1)\right]^{-1} \geq \\ 
			& \qquad\hat{\gamma} \varPhi_{k-d_t-1}^{-T}  \left[ {P}_{k }^i(k-d_t-1)\right]^{-1}{\varPhi_{k-d_t-1} ^{-1}}
		\end{split}
	\end{equation}
	where $\hat{\gamma} \leq 1$. Substituting \eqref{bounded_invertible_trans} into \eqref{bounded_all} satisfies:
	\begin{equation}\label{bounded_k_to_k-1}
		\begin{split}
			&{\left[ {P}_{k}^i(k-d_t)\right]^{ - 1}} \geq \\
			&\qquad \hat{\gamma} \varPhi_{k-d_t-1}^{-T}  \left[ {P}_{k }^i(k-d_t-1)\right]^{-1}{\varPhi_{k-d_t-1} ^{-1}} +\\
			& \qquad \sum\limits_{j \in {N_i}} {{ {\omega}}_{ij}} {(H_{k-d_t}^j)^T}{(R_{k-d_t}^j)^{ - 1}}H_{k-d_t}^j
		\end{split}
	\end{equation}
	According to \eqref{bounded_k_to_k-1}, the relationship between $ {P}_{k }^i(k-d_t)$ and $ {P}_{k }^i(k-d_t-1)$ can be obtained. Therefore, after iteration $d_t$ times for \eqref{bounded_k_to_k-1}, we obtain:
	\begin{equation} \label{bounded_d_times}
		\begin{split}
			&{\left[ {P}_{k }^i(k-d_t)\right]^{ - 1}} \geq  
			\varpi \left(O_{k-d_t,k-2d_t} \right)^{-T}  \left[ {P}_{k-2d_t }^i(k-2d_t)\right]^{-1}\\  
			&\qquad {\left(O_{k-d_t,k-2d_t} \right) ^{-1}}+\sum\limits_{s=0}^{d_t-1} \sum\limits_{j \in {N_i}}\varpi {{ {\omega}}_{ij}}\left( O_{k-d_t,k-d_t-s} \right)^{-T}   \\   
			& \qquad \times {(H_{k-d_t-s}^j)^T}{(R_{k-d_t-s}^j)^{ - 1}}  
			H_{k-d_t-s}^j\left( O_{k-d_t,k-d_t-s} \right)^{-1} 
		\end{split}
	\end{equation}
	where $\varpi = \hat{\gamma}^{d_t}$. Also, iterating $Z$ times for \eqref{bounded_d_times}, it follows that:
	\begin{equation} \label{bounded_L_times}
		\begin{split}
			&{\left[ {P}_{k }^i(k-d_t)\right]^{ - 1}} \geq 
			\varpi^{L} \left[O_{k-d_t, k-d_t-Zd_t} \right]^{-T} \\   
			&\qquad \times\left[ {P}_{k-Zd_t }^i(k-d_t-Zd_t)\right]^{-1}  
			{\left(O_{k-d_t, k-d_t-Zd_t} \right) ^{-1}} \\ 
			&\qquad +\sum\limits_{\sigma=1}^{Z} \sum\limits_{s=0}^{d_t-1} \sum\limits_{j \in \mathcal{V}} \varpi^{\sigma} {{ {\omega}}_{ij}(\sigma)}\left(O_{k-d_t,k-\sigma d_t-s} \right)^{-T} \\   
			& \qquad \times {(H_{k-\sigma d_t-s}^j)^T}{(R_{k-\sigma d_t-s}^j)^{ - 1}}  H_{k-\sigma d_t-s}^j \\ 
			& \qquad \times\left(O_{k-d_t,k-\sigma d_t-s} \right)^{-1} 
		\end{split}
	\end{equation}
	
	Because ${{ {\omega}}_{ij}(\sigma)}$ is the element of matrix $\mathcal{W}^{\sigma}$, and the topology of sensor network is undirected graph, thus,  it has ${{ {\omega}}_{ij}(\sigma)} >0 $ for any $\sigma \geq 1$. Therefore, \eqref{bounded_L_times} can be rewritten as:
	\begin{equation} \label{bounded_delete_cov}
		\begin{split}
			&{\left[ {P}_{k }^i(k-d_t)\right]^{ - 1}} \geq  
			{\tilde{\omega}}_{\text{min}} \varpi^{Z} \sum\limits_{\sigma=1}^{Z} \sum\limits_{s=0}^{d_t-1} \sum\limits_{j \in \mathcal{V}}\left(O_{k-d_t,k-\sigma d_t-s} \right)^{-T}\\ &\qquad{(H_{k-\sigma d_t-s}^j)^T} 
			{(R_{k-\sigma d_t-s}^j)^{ - 1}} H_{k-\sigma d_t-s}^j\left(O_{k-d_t,k-\sigma d_t-s} \right)^{-1} 
		\end{split}
	\end{equation}
	where ${\tilde{\omega}}_{\text{min}} = \text{min}\{ {{ {\omega}}_{ij}(\sigma)} \}$ and $\sigma\geq1$, $i,j \in \mathcal{V}$. Defining $k_z= k-Zd_t-d_t+1$, \eqref{bounded_delete_cov} follows that:
	
	\begin{equation} \label{bounded_kl}
		\begin{split}
			&{\left[ {P}_{k }^i(k-d_t)\right]^{ - 1}}\\
			&\geq 
			{\tilde{\omega}}_{\text{min}} \varpi^{Z}  
			\left(O_{k-d_t,k_z} \right)^{-T} 
			\sum\limits_{s=k_z}^{k_z+\bar{n}+n-1} \sum\limits_{j \in \mathcal{V}}\left(O_{s,k_z} \right)^{-T}\\
			&\qquad \times{(H_{s}^j)^T} {(R_{s}^j)^{ - 1}} H_{s}^j\left(O_{s,k_z} \right)^{-1} \left(O_{k-d_t,k_z} \right)^{-1} 
		\end{split}
	\end{equation}
	
	According to \eqref{completely_observable}, the following equation is satisfied:
	
	\begin{equation} \label{bounded_found_the_doundary}
		\begin{split}
			{\left[ {P}_{k }^i(k-d_t)\right]^{ - 1}} \geq  
			{\tilde{\omega}}_{\text{min}} \varpi^{Z}  \bar{\alpha}
			\left(O_{k-d_t,k_z} \right)^{-T} 
			\left(O_{k-d_t,k_z} \right)^{-1} 
		\end{split}
	\end{equation}

	From Assumption \ref{varPhi_initial_cov_invertible}, and $[\varPhi_{k}]^{-1} \geq \eta I$ is invertible for any $k \in N^{+}$, \eqref{bounded_found_the_doundary} can be rewritten as:
	\begin{equation} \label{bounded_doundary}
		{\left[ {P}_{k }^i(k-d_t)\right]^{ - 1}} \geq  
		{\tilde{\omega}}_{\text{min}} \varpi^{Z}  \bar{\alpha}
		\eta^{2(Zd_t-1)}
	\end{equation}
	When $k \in \{(Z+1)d_t, \dots , +\infty\}$ and $s \in D_t(k) \backslash \{k-d_t \}$, according to \eqref{COV_delay}, \eqref{predict_delay_P}, and Lemma \ref{lemma1_matrix_invert_beta}, we have:
	\begin{equation} \label{bounded_s}
		\begin{split}
			{\left[ {P}_{k }^i(s)\right]^{ - 1}} \geq 
			\hat{\gamma} [\varPhi_{s-1}]^{-T}  \left[ {P}_{k }^i(s-1)\right]^{-1}{[\varPhi_{s-1} ]^{-1}} 
		\end{split}
	\end{equation}
	Similar to previous steps, after iterating the equation \eqref{bounded_s} for $s-k+d_t$ times, it follows that:
	\begin{equation} \label{bounded_s-k+d_t}
		\begin{split}
			{\left[ {P}_{k }^i(s)\right]^{ - 1}} \geq &
			\hat{\gamma}^{s-k+d_t} \left(O_{s,k-d_t} \right)^{-T}    \left[ {P}_{k }^i(k-d_t)\right]^{-1}\\
			& \times{\left(O_{s,k-d_t} \right)^{-1}} 
		\end{split}
	\end{equation}
	Combining \eqref{bounded_doundary} and \eqref{bounded_s-k+d_t}, yields:
	\begin{equation} \label{bounded_D_T(k)}
		\begin{split}
			{\left[ {P}_{k }^i(s)\right]^{ - 1}} &\geq 
			\hat{\gamma} ^{s-k+d_t}{\tilde{\omega}}_{\text{min}} \varpi^{Z}  \bar{\alpha}
			\eta^{2(Zd_t-1)} \left(O_{s,k-d_t} \right)^{-T} \\		   
			&\quad \times{\left(O_{s,k-d_t} \right)^{-1}} \\
			&\geq  {\tilde{\omega}}_{\text{min}} \varpi^{Z}  \bar{\alpha} \hat{\gamma} ^{s-k+d_t} \eta^{2(Zd_t-1)+ 2(s-k+d_t)} \\
			&> 0
		\end{split}
	\end{equation}
	where $k \in \{(Z+1)d_t, \dots , +\infty\}$ and $s \in D_t(k)$. Combined with the study mentioned above, the error covariance is shown to be bounded, that is, ${[ {P}_{k }^i(s)]^{ - 1}} \geq \text{max}\{\vartheta, {\tilde{\omega}}_{\text{min}} \varpi^{Z}  \bar{\alpha} \hat{\gamma} ^{s-k+d_t} \eta^{2(Zd_t-1)+ 2(s-k+d_t)}\}$, for any $k \in  [1, +\infty)$ and $s \in D_t(k)$. This proof is  thus complete.
\end{proof}

\begin{remark}\label{bound_of_delay}
	When filtering accuracy is assured in a vehicle-borne sensor network, we are more concerned with the real permissible delay boundary. Thus, we may reversibly deduce the boundary of delay under the corresponding error covariance using the proof of boundedness of error covariance in the previous section.
	
	From \eqref{bounded_D_T(k)}, we obtain
	\begin{equation}\label{cov_bound}
		{\left[ {P}_{k }^i(s)\right]^{ - 1}} \geq  {\tilde{\omega}}_{\text{min}} \varpi^{Z}  \bar{\alpha} \hat{\gamma} ^{s-k+d_t} \eta^{2(Zd_t-1)+ 2(s-k+d_t)}
	\end{equation}
	Taking the logarithm of both sides of \eqref{cov_bound},  we have
	\begin{equation}\label{delay_bound_process}
		\begin{split}
			\text{In} & \left\{ {\left[ {P}_{k }^i(s)\right]^{ - 1}}  \right\}  \geq \text{In} \left({\tilde{\omega}}_{\text{min}} \varpi^{Z} \bar{\alpha} \right)+(s-k+d_t) \text{In} \hat{\gamma}\\
			& + \big[2(Zd_t-1)+ 2(s-k+d_t)\big] \text{In} \eta
		\end{split}
	\end{equation}
	After sorting out \eqref{delay_bound_process}, the transmission delay boundary is
	\begin{equation}\label{delay_bound}
		d_t \leq\log _{\hat{\gamma} \eta^{2(Z+1)} } \left[  \hat{\gamma} ^{k-s} 
		\eta^{2(k-s+1)}\left({ {P}_{k }^i(s) }{\tilde{\omega}}_{\text{min}} \varpi^{Z}  \bar{\alpha} \right)^{-1} \right]		
	\end{equation}
	
\end{remark}

\subsection{Anti-delay FtDKF algorithm with directed graph}
The proposed filtering algorithm and its demonstration were built on an undirected network in the previous subsection. 
Considering that it is difficult to guarantee ideal undirected communication topology condition in complex traffic environment, the communication topology condition is reduced to directed graph.
The finite-time DKF with transmission delays based on a directed graph is discussed in this subsection to reduce the vehicle-borne sensor communication burden and improve the algorithm's usability in vehicle network.

\begin{definition}\cite{12}\label{directed_graph_definition}
	A directed graph $\mathcal{G}$ is strongly connected if and only if a directed path exists from every sensor node to every other sensor  node. The directed graph of the sensor network diameter is $\bar{d}_g$.  
\end{definition}

\begin{lemma}\cite{he2018consistent}\label{postive_graph_define}
	If the directed graph $\mathcal{G}=(\mathcal{V},\mathcal{E},\mathcal{W})$ is strongly connected for $\mathcal{V} =\{1,2,\dots, n \}$, then all elements of $\mathcal{W}^{s}$, $s \geq n-1$ are positive.
\end{lemma}

\begin{theorem}\label{algorithm_convergence_with_delay_directed}
	Considering the vehicle-borne sensor systems \eqref{system_model} and \eqref{measurement_equation}, assuming that the graph $\mathcal{G}$ is a directed graph, and satisfies Definition \ref{directed_graph_definition}.
	Algorithm \ref{alg1:Framwork} shows the anti-delay DKF, which has the property of finite-time convergence.
\end{theorem}

\begin{proof}\label{proof_of_alg2_directed}
		According to the proof of Algorithms \ref{alg1:Framwork}, the  estimated state convergence of node $i$ can be completed by transmitting the neighbor node ($\mathcal{N}_i$) information to node $i$.
	Similarly, when the communication topology of sensor network is a directed graph, it still needs to complete finite-time convergence through information transmission.  The number of information transfers is also related to the maximum number of hops between two nodes (i.e., the diameter of the directed graph $\bar{d}_g$).
	
	Combined with the previous convergence proof of the undirected graph, when the directed graph is applied in Algorithm \ref{alg1:Framwork}, at least $\bar{d}_g d_t$ iterations are needed to achieve finite-time convergence, where $d_t$ is the maximum transmission delay.
\end{proof}

\begin{remark}\label{directed_cov_boundess}
	When the sensor network is a directed graph, the boundedness of estimation error covariance needs to be further explored. According to the proof in subsection \ref{undirected_boundness} and Lemma \ref{postive_graph_define}, it has ${{ {\omega}}_{ij}(\sigma)} >0 $ for any $\sigma \geq n$. Therefore, for \eqref{bounded_L_times}, we have:
	
	\begin{equation} \label{bounded_delete_cov_directed}
		\begin{split}
			&{\left[ {P}_{k }^i(k-d_t)\right]^{ - 1}} \geq  
			{\tilde{\omega}}_{\text{min}} \varpi^{Z} \sum\limits_{\sigma=n}^{Z} \sum\limits_{s=0}^{d_t-1} \sum\limits_{j \in \mathcal{V}}\left(O_{k-d_t,k-\sigma d_t-s} \right)^{-T}\\ &\qquad{(H_{k-\sigma d_t-s}^j)^T} 
			{(R_{k-\sigma d_t-s}^j)^{ - 1}} H_{k-\sigma d_t-s}^j\left(O_{k-d_t,k-\sigma d_t-s} \right)^{-1} 
		\end{split}
	\end{equation}
	where $Z = n+ \bar{n}$ and ${\tilde{\omega}}_{\text{min}} = \text{min}\{ {{ {\omega}}_{ij}(\sigma)} \}$ and $\sigma>n$, $i,j \in \mathcal{V}$. 
	\begin{equation} \label{bounded_kl_directed}
		\begin{split}
			{\left[ {P}_{k }^i(k-d_t)\right]^{ - 1}} &\geq 
			{\tilde{\omega}}_{\text{min}} \varpi^{Z}  
			\left(O_{k-d_t,k_z} \right)^{-T} 
			\sum\limits_{s=k_z}^{k_z+\bar{n}} \sum\limits_{j \in \mathcal{V}}\left(O_{s,k_z} \right)^{-T}\\
			&\times{(H_{s}^j)^T} {(R_{s}^j)^{ - 1}} H_{s}^j\left(O_{s,k_z} \right)^{-1} \left(O_{k-d_t,k_z} \right)^{-1} \\
			&\geq {\tilde{\omega}}_{\text{min}} \varpi^{Z}  \alpha
			\left(O_{k-d_t,k_z} \right)^{-T} 
			\left(O_{k-d_t,k_z} \right)^{-1} 
		\end{split}
	\end{equation}
	The initial number of $\sigma$ is the major distinction between proofs of directed and undirected graphs. Thus, the error covariance boundary value is explicitly given as:
	
	\begin{equation} \label{bounded_D_T(k)_directed}
		\begin{split}
			{\left[ {P}_{k }^i(s)\right]^{ - 1}} &\geq  {\tilde{\omega}}_{\text{min}} \varpi^{Z}  \alpha \hat{\gamma} ^{s-k+d_t} \eta^{2(Zd_t-1)+ 2(s-k+d_t)} \\
			&> 0
		\end{split}
	\end{equation}  
	Since $\bar{\alpha} > \alpha$, the bound of error covariance in \eqref{bounded_D_T(k)_directed} is greater than \eqref{bounded_D_T(k)}. 
\end{remark}

\begin{remark}\label{main_application}
Based on this analysis, the anti-delay DKF algorithm suggested in this paper can be applied to undirected and directed graphs.
Because the real vehicle information fusion environment is diverse and unknown, the suggested algorithm offers a broader range of applications in vehicle network.
\end{remark}
\begin{figure*}[t]
	\centering
	\includegraphics[width=6in]{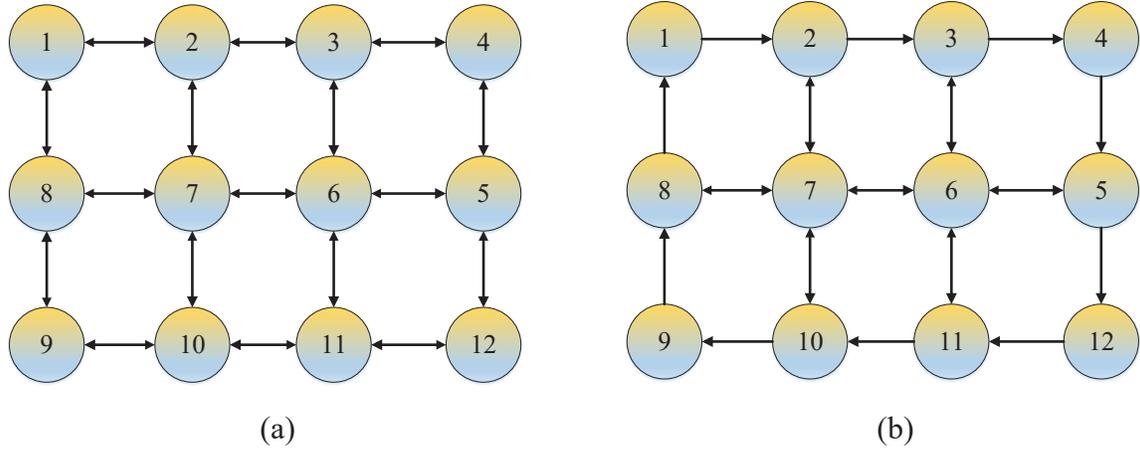}
	\caption{Structure of the sensor network. (a) Case 1: undirected graph (b) Case 2: directed graph}
	\label{fig_network}
\end{figure*}

\section{Simulations and experiment}\label{Numerical experiments}

The discrete-time system can be stated as follows for tracking systems with twelve sensors:
\begin{equation}\label{Numerical_system_model}
	x_{k+1} = 	
	\begin{bmatrix}
		 & 1 &T  & T^2/2    \\
		& 0 & 1 & T \\
		& 0 & 0 & 1 \\
	\end{bmatrix} \ \  x_{k} + w_{k}
\end{equation}

where $T =0.01$ s is the sampling period and the target state $x_{k} = [s_{k}, \dot{s}_{k}, \ddot{s}_{k}]^T$, and $s_{k}, \dot{s}_{k}, \ddot{s}_{k}$ are the position, velocity and acceleration of the target, respectively. 
\begin{equation}\label{sensors_equation}
	y_{k}^{i} = {H_{k}^{i}} x_{k} + v_{k}^{i}, \qquad i = 1, 2, \dots, 12
\end{equation}
where ${H^1}=[1,0,0]$, ${H^2}=[0,1,0]$, ${H^3}=[0,0,1]$, ${H^4}=[1,0,0]$, ${H^5}=[0,1,0]$, ${H^6}=[0,0,1]$, ${H^7}=[1,0,0]$, ${H^8}=[0,1,0]$, ${H^9}=[0,0,1]$, ${H^{10}}=[1,0,0]$, ${H^{11}}=[0,1,0]$ and ${H^{12}}=[0,0,1]$; $Q=1$ and $R_i= \textrm{diag}[0.8, 1.0, 2.0, 1.0, 0.5, 1.5, 1.0, 1.0, 0.1, 1.0, 1.5, 1.0]$ for any instant. The initial values of state are $\mu_0= [0, \dots, 0]^T$ and $P_0= I_{12}$, respectively. The communication network between sensors is shown in Fig \ref{fig_network}. The elements of the weighted matrix $\mathcal{W}$ are $\omega_{ij}= 1/| \mathcal{N}_i|$, where $j\in \{ \mathcal{N}_i\}$. The time-varying transmission delays $d^t_{ij}(k)$ satisfy $\{d^t_{ij}(k)|d^t_{ij}(k) \in \{0,1,\dots,d_{t}\}$, let $d_{t}=$4. The probability of occurrence of each delay value is the same as $1\backslash (d_{t}+1)=$ 0.2.

The simulation is divided into two parts.  First,  the tracking effects of the proposed Algorithm \ref{alg1:Framwork} and that from  \cite{2020Distributed}  are compared in the case of the same target motion.
Second, the adaptability of Algorithm \ref{alg1:Framwork} in undirected and directed graphs is verified, and robustness for different transmission delay bounds is verified.

\subsection{Accurate tracking effect of Algorithm \ref{alg1:Framwork} on target}

Transmission delays are unavoidable in a vehicle-borne sensor network due to sensor performance discrepancies and external interference.
Thus, the proposed Algorithm \ref{alg1:Framwork} takes the transmission delays of the sensor network into account.

The tracking effect comparison between Algorithm \ref{alg1:Framwork} and Reference \cite{2020Distributed} is provided in Fig. \ref{fig_alg2_delay} to emphasize the proposed algorithm's great performance when the bound of transmission delays $d_t$.
Under the same initial conditions, both Algorithm \ref{alg1:Framwork} and Reference \cite{2020Distributed} can track the target. However, careful examination reveals that the suggested Algorithm \ref{alg1:Framwork}'s tracking accuracy is better than that of \cite{2020Distributed}. As a result, Algorithm \ref{alg1:Framwork} performs better in terms of tracking.

Correspondingly, Fig. \ref{fig_alg2_compare_mse} shows the MSE of Algorithm \ref{alg1:Framwork} and Reference \cite{2020Distributed}. Consistent with the above analysis, because Algorithm \ref{alg1:Framwork} achieves more accurate tracking performance, the MSE of Algorithm \ref{alg1:Framwork} is less than that of \cite{2020Distributed} in terms of position, velocity and acceleration.

\begin{figure}[!t]
	\centering
	\includegraphics[width=3.8in]{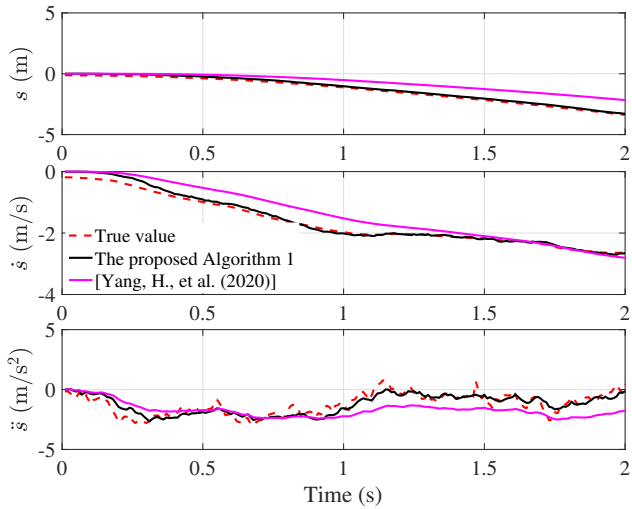}
	\caption{Comparison of the tracking effect between Algorithm \ref{alg1:Framwork} and Reference \cite{2020Distributed}  with transmission delay $d_t =3$ for case 1.}
	\label{fig_alg2_delay}
\end{figure}

\begin{figure}[!t]
	\centering
	\includegraphics[width=3.8in]{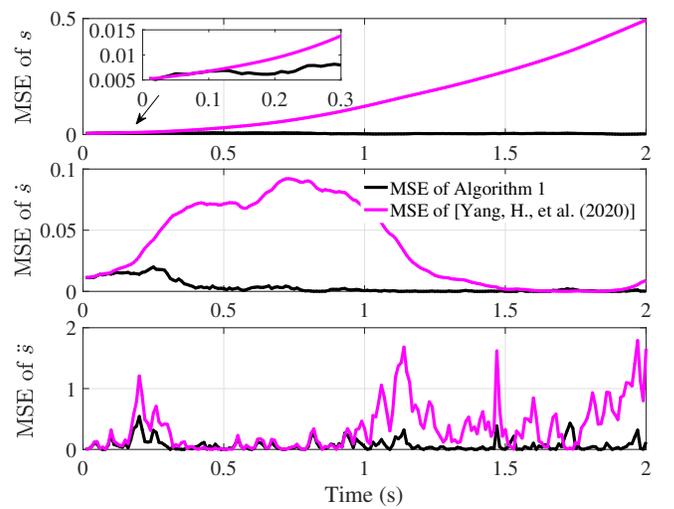}
	\caption{MSE of Algorithm \ref{alg1:Framwork} and Reference \cite{2020Distributed}  with transmission delay $d_t =3$ for case 1.}
	\label{fig_alg2_compare_mse}
\end{figure}

\subsection{Performance verification of Algorithm \ref{alg1:Framwork} under undirected graph and directed graph}

To reduce communication burden, the suggested algorithm is more relevant in practical traffic scenarios. In this subsection, we focus on the tracking effect of the same target by Algorithm \ref{alg1:Framwork} for undirected  and directed graph, respectively. Fig. \ref{fig_un_directed_compare_mse} shows the MSE of the tracking effect. 
The centralized algorithm is contrasted to examine the tracking effect of the proposed Algorithm \ref{alg1:Framwork} on the target.
Algorithm \ref{alg1:Framwork} can track the target accurately 
whether it is used with a directed or undirected graph. The MSE in Fig. \ref{fig_un_directed_compare_mse} agree with these results.

Through careful observation, with the same target, when Algorithm \ref{alg1:Framwork} is applied to the directed graph, its MSE is greater than that of the undirected graph because when the sensor network's communication topology is a directed graph, the number of iterations required to complete finite-time convergence is $\bar{d}_g d_t$, which is more than the number required for an undirected graph. 
In addition, when compared to an undirected graph, the quantity of data that can be sent at the same time is limited. Thus, MSE will be greater when the algorithm is used on directed graph.

\begin{figure}[!t]
	\centering
	\includegraphics[width=3.8in]{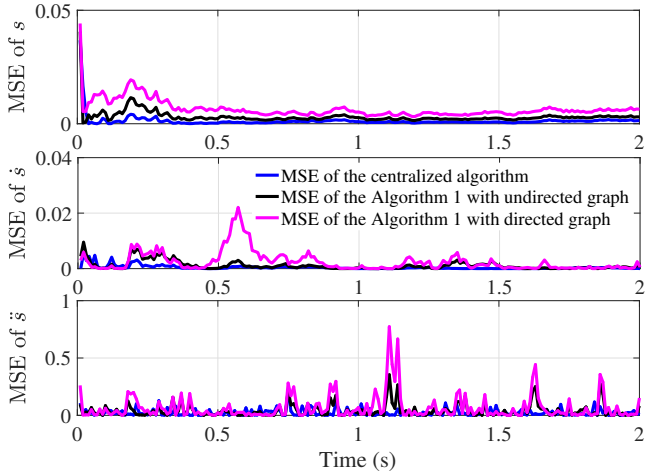}
	\caption{MSE of the Algorithm \ref{alg1:Framwork} for undirected (case 1) and directed (case 2) graphs.}
	\label{fig_un_directed_compare_mse}
\end{figure}

\subsection{Robustness of Algorithm \ref{alg2:Framwork} to different transmission delay  bounds}

By varying the upper bound of transmission delays, the robustness of the suggested Algorithm \ref{alg1:Framwork} can be investigated in more detail. Fig.\ref{fig_compare_delay} compares different delay bounds and shows that as transmission delay rises, the MSE of the suggested Algorithm \ref{alg1:Framwork} grows moderately. 
These MSEs all tend to be stable and diminish after a short number of iterations, which is due to the suggested algorithm's use of finite-time control technology.
\begin{figure}[!t]
	\centering
	\includegraphics[width=3.8in]{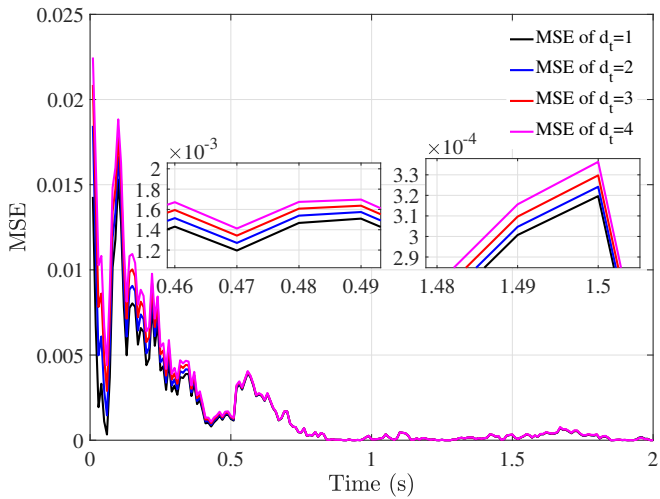}
	\caption{The position MSE of the proposed algorithm subject to different transmission delay  bounds for case 1.}
	\label{fig_compare_delay}
\end{figure}

\subsection{Experimental case-study}
To highlight the feasibility of the Algorithm \ref{alg1:Framwork} in practical  traffic scenarios, as shown in Fig. \ref{Experimental case-study}, three radar sensors are used to track a moving car target in this subsection. The car moves along the approximate straight line, starting from the initial position (0, 6.123)[m] with nearly-constant speed. Three 24G millimeter wave radar sensors (Nanoradar SP25 millimeter wave radar) are located in (0,0) to detect the relative distance and azimuth between the car and the radar sensors. 

The experiment consists of 21 samples with sampling interval $T$=1/100s. The obtained experimental results  are shown in Fig. \ref{Tracking result}, and the tracking error MSE is shown in Fig. \ref{Tracking MSE}.  

According to the tracking results shown in  Fig. \ref{Tracking result} and Fig. \ref{Tracking MSE}, it can be seen that the Algorithm \ref{alg1:Framwork} proposed in this paper has similar tracking effect with the centralized algorithm, and better than Reference \cite{2020Distributed}. The above tracking results are consistent with the previous simulation. Moreover, the excellent performance of the proposed algorithm is verified again.

\begin{figure}[!t]
	\centering
	\includegraphics[width=3.5in]{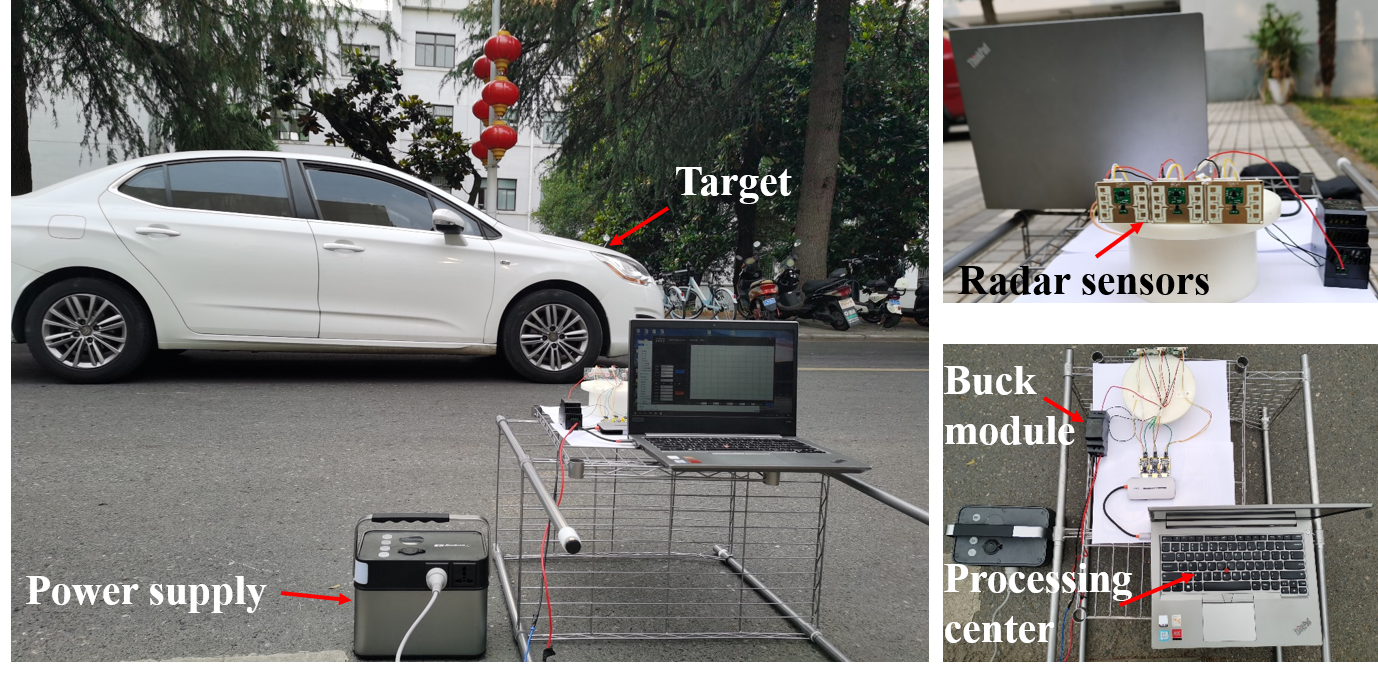}
	\caption{Experimental setup.}
	\label{Experimental case-study}
\end{figure}

\begin{figure}[!t]
	\centering
	\includegraphics[width=3.8in]{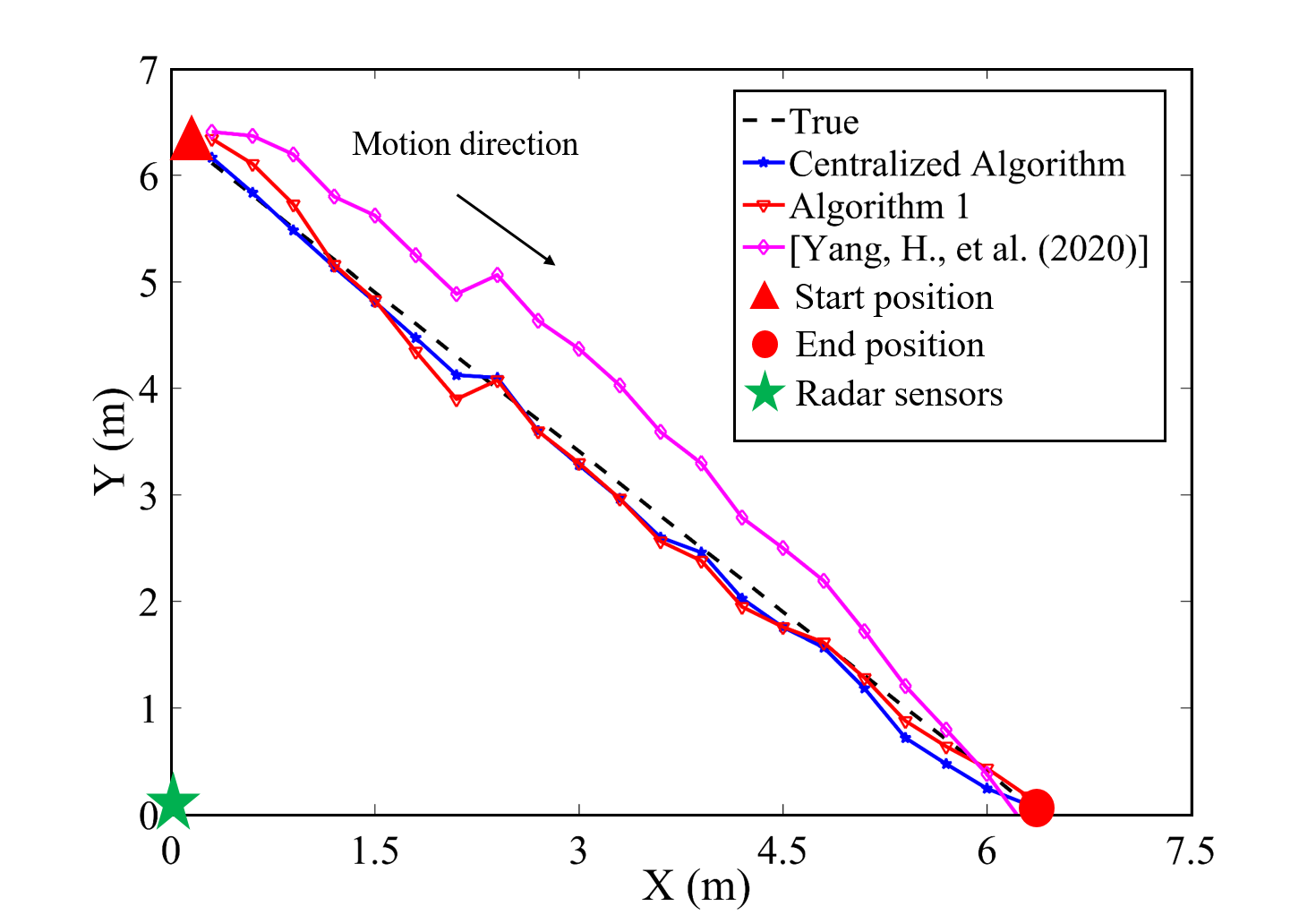}
	\caption{Experimental result.}
	\label{Tracking result}
\end{figure}

\begin{figure}[!t]
	\centering
	\includegraphics[width=3.8in]{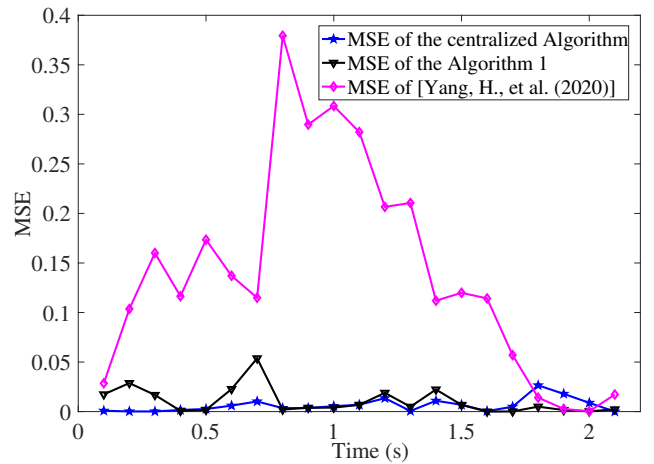}
	\caption{Position MSE with multi-sensor fusion.}
	\label{Tracking MSE}
\end{figure}

\section{Conclusion}\label{conclusion}
This paper investigates the problem of coordinate fusion of  target vehicle state for vehicle-borne sensor networks with time-varying transmission delays. Compared with previous literature, we consider more non-ideal communication constraints in vehicle network, including communication weight, transmission delay, and network topology. Combining the finite-time control strategy and buffer structure, an anti-delay DKF algorithm is designed, with rigorous proof of  simultaneous target state estimation error convergence and anti-delay convergence in the finite-time step.
The ingenious trick is that we use distributed design and optimal vector weight technology to address the high complexity problem caused by direct integration, making the algorithm easy to implement in practical vehicle actual states tracking. 
Furthermore, the upper bound of estimation error covariance is established by using the inverse matrix theorem and Gramian matrix characteristics. Based on this result, the maximum possible delay under the allowable detection accuracy can be determined.
In addition, both the finite-time convergent DFK algorithm and the derivation of the error covariance bound are generalized to strongly connected digraph communication topology, making it more valuable for practical target vehicle states tracking with worse communication conditions.
Finally, the effectiveness and feasibility of the proposed algorithm are verified by simulations and experiment.

\section*{Acknowledgment}
This study was supported in part by the Postgraduate Research \& Practice Innovation Program of Jiangsu Province KYCX20\_0316, foundation strengthening plan technical field fund 2021-JCJQ-JJ-0597.

\bibliographystyle{IEEEtran}
\bibliography{IEEEexample}

\begin{IEEEbiography}[{\includegraphics[width=1in,height=1.25in,clip,keepaspectratio]{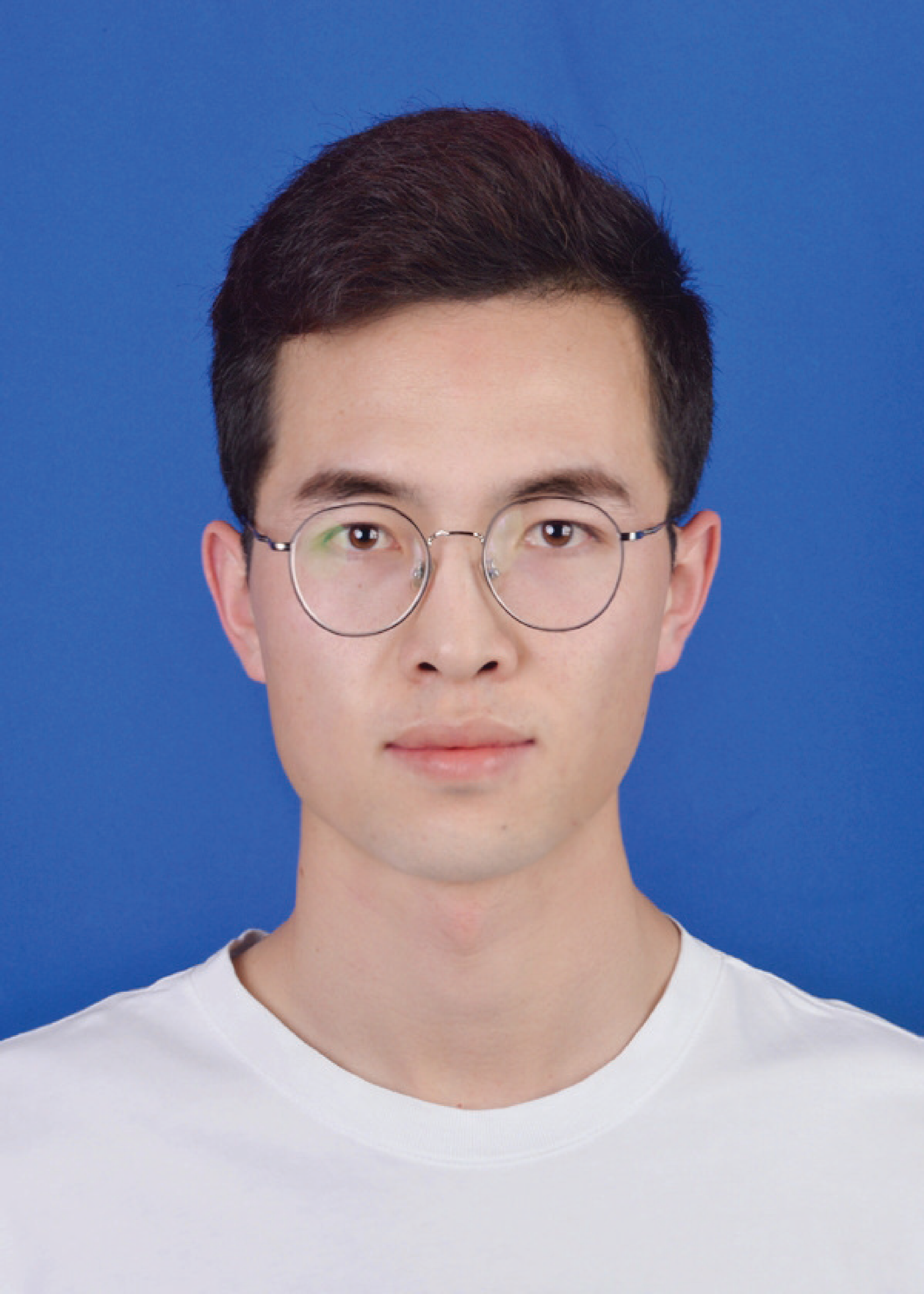}}]{Hang Yu}
	was born in Ningxia, China. He is currently pursuing the Ph.D. degree with the ZNDY of Ministerial Key Laboratory, Nanjing University of Science and Technology, Nanjing, China. His current research interests include networked control systems, cooperative control and intelligent fuze.
\end{IEEEbiography}
\vspace{11pt}

\begin{IEEEbiography}[{\includegraphics[width=1in,height=1.25in,clip,keepaspectratio]{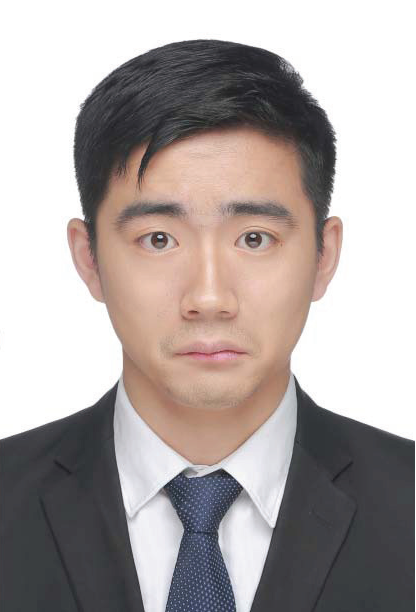}}]{Keren Dai}
	received the B.E. and Ph.D. degrees from Tsinghua University, Beijing, China, in 2014 and 2018, respectively. He is currently an Associate Professor with the School of Mechanical Engineering, Nanjing University of Science and Technology. His research interests include system modeling and simulation, signal processing, cooperative control, networked control systems, power management systems, and micro power devices.
\end{IEEEbiography}
\vspace{11pt}

\begin{IEEEbiography}[{\includegraphics[width=1in,height=1.25in,clip,keepaspectratio]{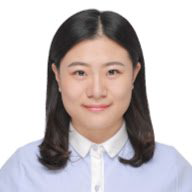}}]{Qingyu Li}
received the B.E. and ph.D. degrees from the Department of Electronics Engineering, Tsinghua University, Beijing, China, in 2014 and 2019  respectively. She is now an engineer at the North Information Control Research Academy Group Company, Ltd., Nanjing, China. Her research interests include communication systems, IoT nodes and networks, information theory and signal processing.
\end{IEEEbiography}
\vspace{11pt}

\begin{IEEEbiography}[{\includegraphics[width=1in,height=1.25in,clip,keepaspectratio]{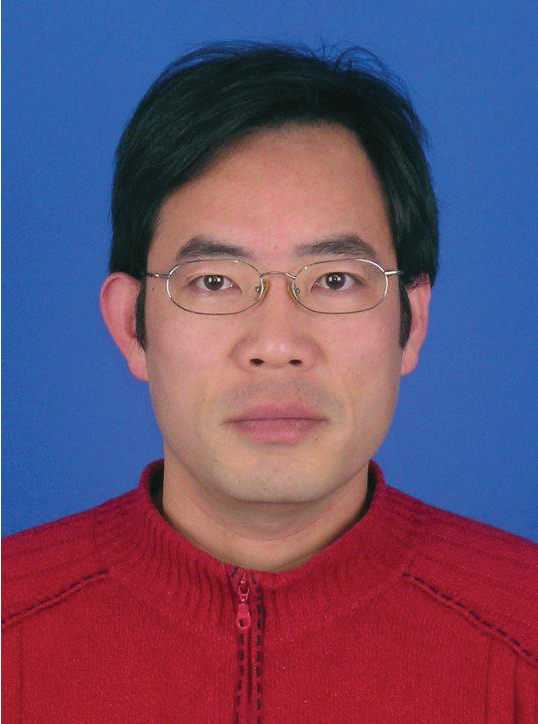}}]{Haojie Li}
	was born in Shanxi, China. He received the Ph.D. degree in mechatronic engineering from the Nanjing University of Science and Technology, Nanjing, China, in 2006, where he is currently a Professor with the School of Mechanical Engineering. His current research interests include intelligent fuze, nonlinear control, detection guidance and control technology, and attitude control technology of floating platforms on water.
\end{IEEEbiography}
\vspace{11pt}

\begin{IEEEbiography}[{\includegraphics[width=1in,height=1.25in,clip,keepaspectratio]{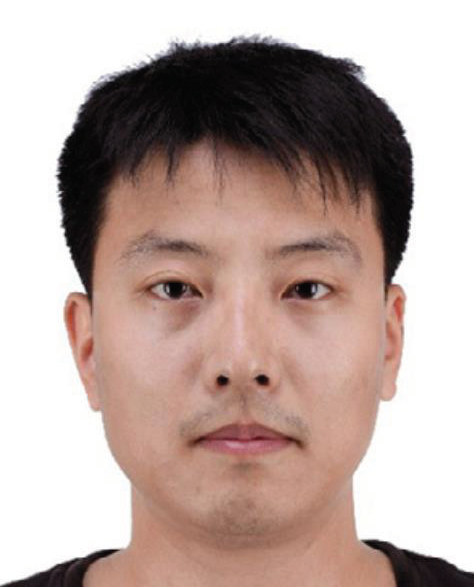}}]{Yao Zou}
	received a B.S. degree in Automation from Dalian University of Technology (DUT), Dalian, China, in 2010, and a Ph.D. degree in Control Science and Engineering from Beihang University (BUAA, formerly named Beijing University of Aeronautics and Astronautics) Beijing, China, in 2016.
	
	He was a Post-Doctoral Research Fellow with the Department of Precision Instrument, Tsinghua University, Beijing, from 2017 to 2018. He is currently a Professor with the School of Automation and Electrical Engineering, University of Science and Technology Beijing, Beijing. His current research interests include nonlinear control, unmanned aerial vehicle control, networked control systems and multiagent control.
\end{IEEEbiography}
\vspace{11pt}

\begin{IEEEbiography}[{\includegraphics[width=1in,height=1.25in,clip,keepaspectratio]{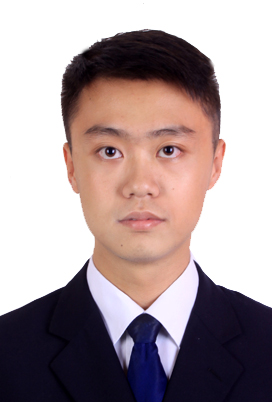}}]{Xiang Ma}
	was born in Neimenggu, China. He is currently pursuing the Ph.D. degree with the ZNDY of Ministerial Key Laboratory, Nanjing University of Science and Technology, Nanjing, China. His current research interest is  cooperative control, networked control systems.
\end{IEEEbiography}
\vspace{11pt}

\begin{IEEEbiography}[{\includegraphics[width=1in,height=1.25in,clip,keepaspectratio]{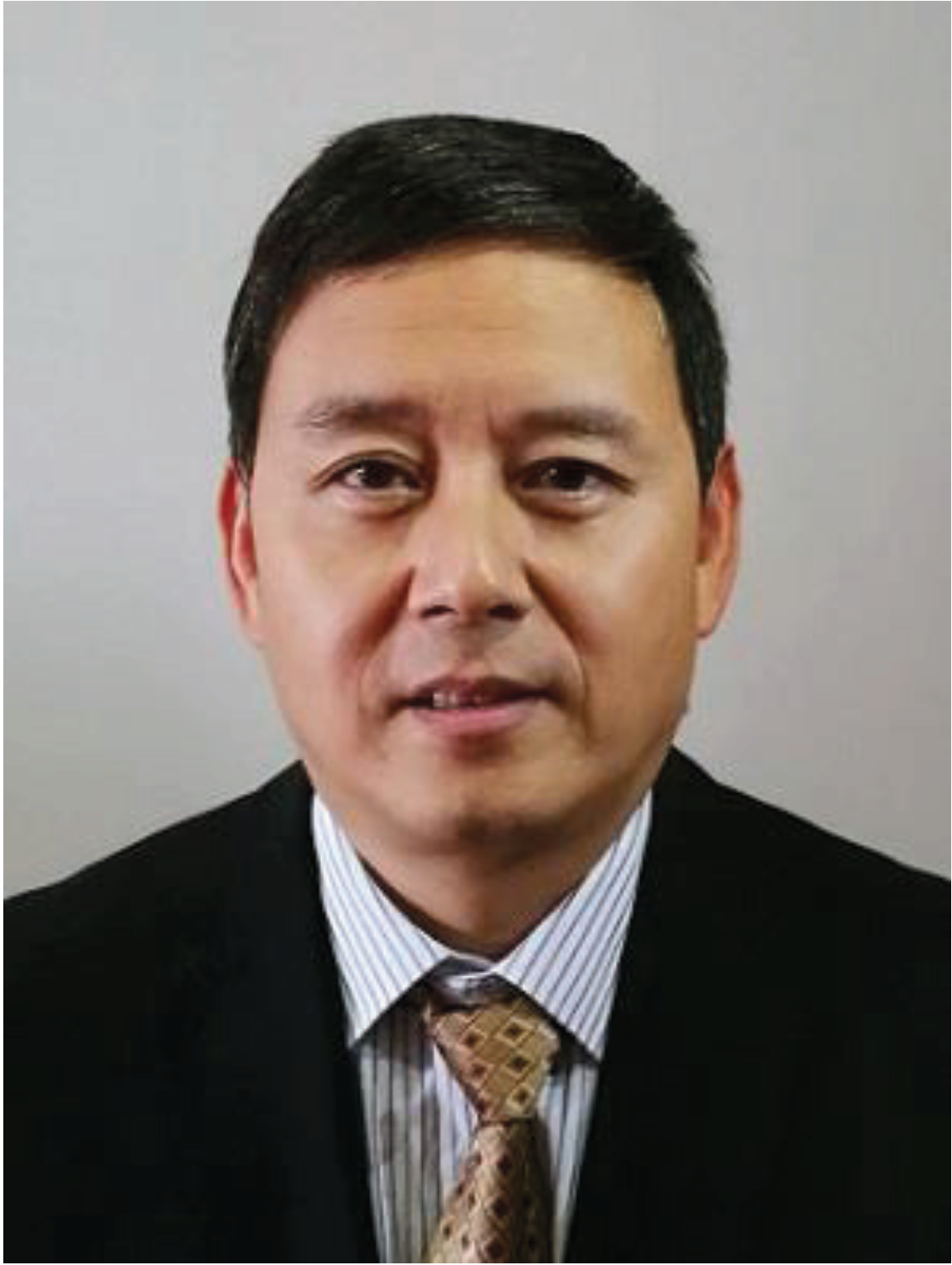}}]{Shaojie Ma}
	was born in Henan, China. He received the Ph.D. degree in mechatronic engineering from the Nanjing University of Science and Technology, Nanjing, China, in 2006, where he is currently a Professor with the School of Mechanical Engineering. His current research interests include intelligent fuze, electromechanical system analysis and design technology and nonlinear control.
\end{IEEEbiography}
\vspace{11pt}

\begin{IEEEbiography}[{\includegraphics[width=1in,height=1.25in,clip,keepaspectratio]{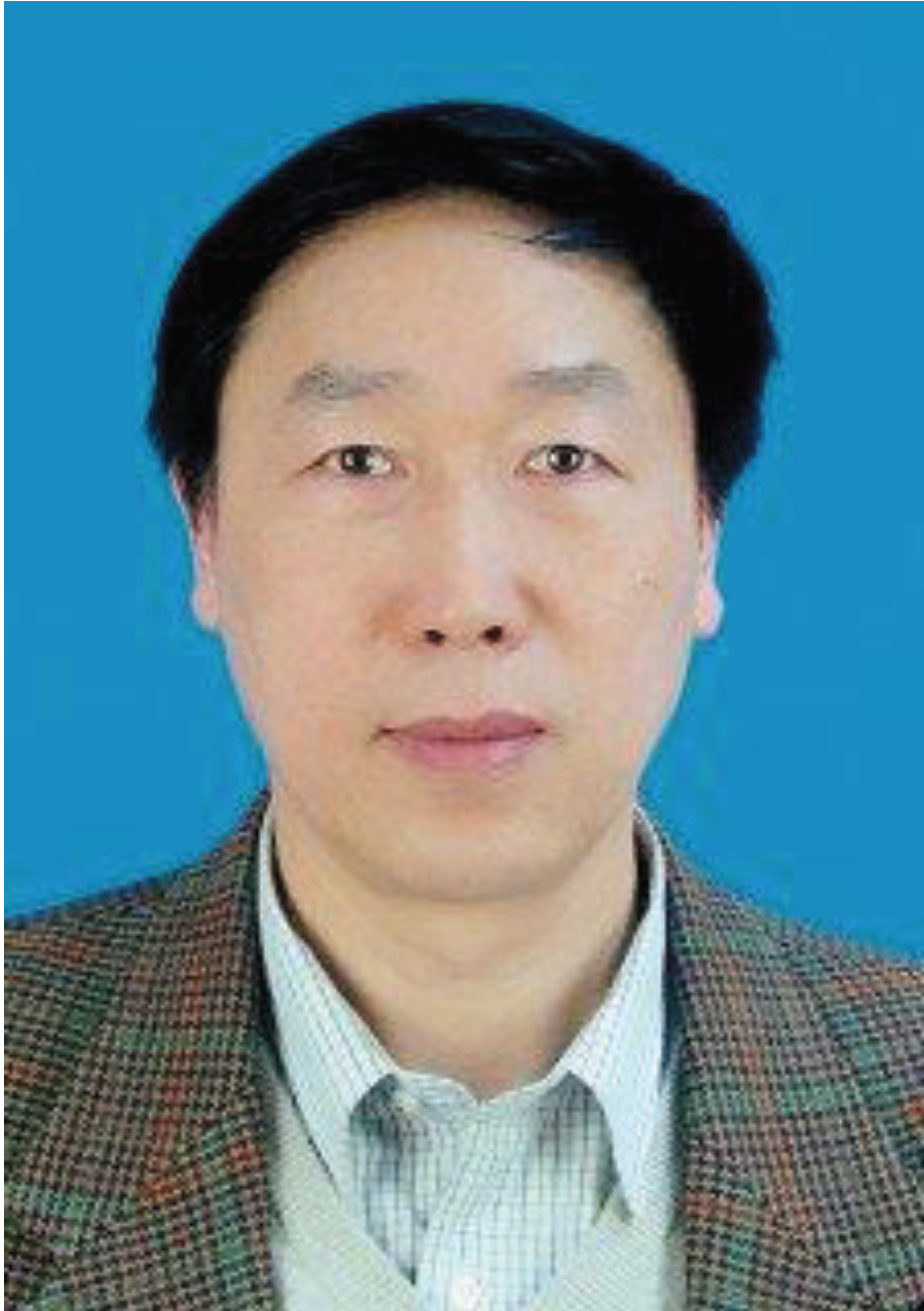}}]{He Zhang}
	was born in Henan, China. He received his Ph.D degree in Measurement Technology and Instruments from Nanjing University of Aeronautics and Astronautics, Nanjing, China. He is currently a Professor in School of Mechanical Engineering, Nanjing University of Science and Technology, Nanjing, China. His research interests include networked control systems, nonlinear control, mechatronics and weapon system applications. Professor ZHANG is the director of the Institute of Mechanical and Electrical Engineering of NJUST, the associate director of the ZNDY National Defense Key Laboratory, and the editorial board of the Journal of Detection and Control.
\end{IEEEbiography}

\vfill

\end{document}